\documentclass[11pt,letterpaper]{article}

\usepackage{amsmath,amsthm,nicefrac}
\usepackage{amsfonts, amstext}
\usepackage{bm}

\usepackage{subcaption}
\usepackage{hhline}
\usepackage{soul}
\usepackage[table,dvipsnames]{xcolor}

\usepackage{typearea}
\typearea{14}
\usepackage[font={small,it}]{caption}

\usepackage{setspace}

\usepackage[shortlabels]{enumitem}
\usepackage[breaklinks]{hyperref}
\hypersetup{colorlinks=true,
            citebordercolor={.6 .6 .6},linkbordercolor={.6 .6 .6},
citecolor=blue,urlcolor=blue,linkcolor=blue,pagecolor=black,breaklinks}

\usepackage[nameinlink,capitalize]{cleveref}
\Crefname{algocf}{Algorithm}{Algorithms}
\crefname{algocfline}{line}{lines}
\Crefname{invariant}{Invariant}{Invariants}
\Crefname{claim}{Claim}{Claims}
\Crefname{corollary}{Corollary}{Corollaries}
\Crefname{subclaim}{Subclaim}{Subclaims}

\usepackage{epsfig}
\usepackage{amsthm,amssymb}

\usepackage{mathrsfs}
\usepackage{xspace}
\usepackage{soul}
\usepackage{latexsym}
\usepackage{bbm}
\usepackage{dsfont}

\usepackage{accents}

\usepackage{framed}

\definecolor{DarkGray}{rgb}{0.66, 0.66, 0.66}
\definecolor{DarkPowderBlue}{rgb}{0.0, 0.2, 0.6}
\definecolor{fluorescentyellow}{rgb}{0.8, 1.0, 0.0}
\definecolor{cerulean}{rgb}{0.0, 0.48, 0.65}
\definecolor{bleudefrance}{rgb}{0.19, 0.55, 0.91}

\usepackage[ruled,vlined,algonl]{algorithm2e}
\SetEndCharOfAlgoLine{}
\SetKwComment{Comment}{\footnotesize$\triangleright$\ }{}

\SetCommentSty{mycommfont}

\usepackage{thmtools,thm-restate}

\allowdisplaybreaks

\newtheorem{theorem}{Theorem}[section]
\newtheorem{lemma}[theorem]{Lemma}

\newtheorem{observation}[theorem]{Observation}

\theoremstyle{definition}
\newtheorem{defn}[theorem]{Definition}

\theoremstyle{remark}

\usepackage{tikz}
\usepackage{pgfplots}
\pgfplotsset{compat=newest}

\tikzset{
vertex/.style={circle, draw, fill=black, black, inner sep=0pt, minimum width=9pt},
asg/.style={line width=1pt},
}

\usepackage{comment}

\usepackage[
 backend=biber,
 style=alphabetic,
 citetracker,
 hyperref=auto,
 maxcitenames=5,
 sortcites,
 sorting=nyt,
 maxbibnames=12,
 date=year,
 isbn=false,
 url=false,
 doi=false,
 eprint=false,
]{biblatex}

\newbibmacro{string+doiurlisbn}[1]{
  \iffieldundef{doi}{
    \iffieldundef{url}{
      #1
    }{
      \href{\thefield{url}}{#1}
    }
  }{
    \href{http://dx.doi.org/\thefield{doi}}{#1}
  }
}

\DeclareFieldFormat
[article,inbook,incollection,inproceedings,patent,thesis,unpublished]
  {title}{\usebibmacro{string+doiurlisbn}{#1}}

\usepackage{titlesec}
\usepackage{lipsum}
\usepackage{float}
\titlespacing{\paragraph}{
  0pt}{
  0.3\baselineskip}{
  1em}

\usepackage[normalem]{ulem}

\newcommand{\eps}{\varepsilon}

\newcommand{\factor}{\smash{\ceil{\nicefrac{1}{\varepsilon}}}}

\newcommand{\oneminusA}{\eps}

\newcommand{\threshold}{\frac{\varepsilon}{1-\varepsilon}\xspace}

\newcommand{\ALG}{{\ensuremath{\mathrm{SLF}}}\xspace}
\newcommand{\SLF}{\ALG}

\newcommand{\alg}{\ALG}
\newcommand{\opt}{\ensuremath{\mathrm{OPT}}\xspace}
\newcommand{\OPT}{\opt}

\newcommand{\vol}{\operatorname{vol}}

\newcommand{\ceil}[1]{\left\lceil #1 \right\rceil}

\renewcommand{\emptyset}{\varnothing}

\usepackage{xcolor}
\definecolor{job1}{HTML}{e69f00}
\definecolor{job2}{HTML}{009e74}
\definecolor{job3}{HTML}{56b3e9}
\definecolor{job4}{HTML}{cc79a7}
\definecolor{job5}{HTML}{f0e442}
\definecolor{job6}{HTML}{d55e00}
\definecolor{job7}{HTML}{0071b2}

\newcommand\machine{}
\def\machine[#1](#2,#3)(#4,#5){
  \draw[black,#1,line width=2pt] (#4,#3) -- (#2,#3) -- (#2,#5) -- (#4,#5);
}

\newcommand\interval{}
\def\interval[#1](#2,#3){
  \draw[line width=1pt,dotted,black,#1] (#2,0) -- (#2,#3);
}

\newcommand\rectjob{}
\def\rectjob[#1](#2,#3)(#4,#5){
  \draw[black,fill=#1,line width=0.7pt] (#2,#3) rectangle ++(#4,#5);
}

\newcommand\rectjobT{}
\def\rectjobT[#1](#2,#3)(#4,#5)(#6){
  \draw[black,fill=#1,line width=0.7pt] (#2,#3) rectangle ++(#4,#5);
  \node[#1,black] at (#2+#4/2.0,#3+#5/2.0) {{#6}};
}

\newcommand\matchingjob{}
\def\matchingjob(#1,#2)(#3,#4){
  \draw[black,fill=none,line width=1.2pt] (#1,#2) rectangle ++(#3,#4);
}

\newcommand\partialjob{}
\def\partialjob(#1,#2)(#3,#4){
  \draw[black,dashed,fill=none,line width=1.2pt] (#1,#2) rectangle ++(#3,#4);
}

\newcommand\rectjobnob{}
\def\rectjobnob[#1](#2,#3)(#4,#5){
  \draw[white,fill=#1,line width=0.0pt] (#2,#3) rectangle ++(#4,#5);
}

\title{A Simpler Analysis for $\eps$-Clairvoyant Flow Time Scheduling}

\author{
   Anupam Gupta\thanks{
       New York University, \texttt{anupam.g@nyu.edu}. 
   } \and Haim Kaplan\thanks{
       Tel Aviv University,
       \texttt{haimk@post.tau.ac.il}
   }  \and Alexander Lindermayr\thanks{
       Institut für Mathematik, Technische Universität Berlin,
       \texttt{alexander.lindermayr@tu-berlin.de}
   }  \and Jens Schlöter\thanks{
       Centrum Wiskunde \& Informatica (CWI),
       \texttt{jens.Schloter@cwi.nl}
   } \and
        Sorrachai Yingchareonthawornchai\thanks{
       Institute for Theoretical Studies, ETH Zürich,
       \texttt{sorrachai.yingchareonthawornchai@eth-its.ethz.ch}
   }}

\date{}
\begin{document}

\maketitle

\begin{abstract}
We simplify the proof of the optimality of the Shortest Lower-Bound First (SLF) algorithm, introduced by Gupta, Kaplan, Lindermayr, Schlöter, and Yingchareonthawornchai [FOCS'25],  for minimizing the total flow time in the $\eps$-clairvoyant setting.
\end{abstract}

\section{Introduction}

We consider the classical problem 
of minimizing the total flow time 
in the online setting where jobs arrive over time
and preemption is allowed. The flow time of a job is the total time
that a job spends in the system. We study the objective
of minimizing the total flow time of all the jobs on a single machine.

In \cite{GuptaKLSY25}, the authors of the present paper
studied this problem
in the $\eps$-clairvoyant setting, which is defined as follows. 
When a job $j$ arrives, we know nothing about its processing time $p_j$. However, the value $p_j$
is revealed when an $\eps$-fraction of job $j$ remains to be processed, for a constant $\eps \in [0,1]$. 
They showed that a natural algorithm, called Shortest Lower-Bound First (SLF), achieves a competitive ratio of at most $\factor$ when $\eps > 0$, and this bound is best possible for deterministic algorithms.

In this technical note, we present a simplified analysis
of SLF, which yields the same tight bound on the competitive ratio. For the remaining aspects of this model and its
connection to related work, we refer to \cite{GuptaKLSY25}.
In the remainder of this note, we first present the SLF algorithm and then prove the following theorem.
\begin{theorem}[\cite{GuptaKLSY25}]\label{thm:main}
	For any constant $\eps \in (0,1]$, there is a deterministic
	$\eps$-clairvoyant algorithm SLF that is $\factor$-competitive
	for the objective of minimizing the total flow time on a single machine.
\end{theorem}

\subsection{Technical Overview}

We prove that SLF is locally $\factor$-competitive at any time $t$,
meaning that there are at most $\factor$ times as many jobs active in the algorithm's schedule
as in an optimal schedule. This is a necessary property for constant competitive algorithms
for this problem, and we also use it in \cite{GuptaKLSY25}.
However, there are two major places in which we improve our analysis in \cite{GuptaKLSY25}, 
which we explain in the following two paragraphs.

\paragraph{Frozen Jobs instead of Early Arrivals.}
A key ingredient in the analysis of \cite{GuptaKLSY25} is the Early Arrival Lemma, 
which shows that we can make certain jobs arrive earlier
without affecting the set of active jobs at time $t$ or their elapsed times 
in the algorithm's schedule. 
Clearly, such a modification can only help an optimal 
solution to reduce its number of active jobs at some target time $t$, hence never improves the local competitive ratio.
In \cite{GuptaKLSY25}, we iteratively search for what we called \emph{magical} times $s$ with the property that during $[s,t]$ the algorithm does not work on jobs~$j \in J_s$ whose remaining time
is more than $\eps p_j$ at time~$s$ ($p_j$ is the total length of job $j$).  
Then, we move the arrival time of jobs  
that arrive after time~$s$ to arrive earlier at time~$s$. 
This modification can only help an optimal solution to reduce its load at time $t$ and adds structure to the schedule produced by
our algorithm from time $s$ onwards, which greatly helps in the analysis. 

In the analysis presented here, we completely avoid this technical reduction. Instead, we mark such jobs $J_s$ as \emph{frozen}
as soon as the algorithm no longer works on them, before time $t$. 
This is possible due to a more powerful Fast Forward Lemma (cf.\ \Cref{lem:fastforward}), 
which does not require all jobs to arrive at the beginning of the fast-forward interval, but instead allows jobs to arrive 
over time until the next job becomes frozen or all jobs are known.

\paragraph{Fixed-Length Prefix Bounds instead of Valid Assignments.}
The second major change is that we use a weaker invariant over time.
In \cite{GuptaKLSY25}, we inductively show for any time $t$ a fractional assignment
between the active jobs
in the algorithm's schedule and the active jobs in the optimal schedule.
The value of a node corresponding to job $j$, in this fractional matching,
is its remaining processing time at time $t$ in the algorithm's schedule or the optimal schedule, respectively (we call this the job's volume at time $t$).
Specifically, we show
 that at any time there exists an assignment that is at most $\factor$-expanding.\footnote{This means that any prefix of length $x$ of SLF's queue sorted in nonincreasing order of remaining time is assigned to  at most   $\factor \cdot x$ jobs of $\OPT$.}
 This
allows us to bound the number of active jobs in the algorithm's schedule
with $\factor$ times the number of active jobs in the optimal schedule, thus proving \Cref{thm:main}.
This invariant has also been used in Schrage's original proof of SRPT's optimality for minimizing total flow time~\cite{Schrage68}.

We use a simpler invariant instead.
As in \cite{GuptaKLSY25} we fix
 a time $t$ for which we want to prove local competitiveness, 
and then only show an invariant (that depends on $t$) over all times $t' \leq t$. But
instead of requiring a bounded expansion for all prefixes of active jobs in the optimal solution, ordered by
 decreasing volume as in \cite{GuptaKLSY25},
we only require that the total volume of the $\factor \cdot |\opt(t',t)|$ jobs with largest volume 
in the algorithm's queue is at least as much as the total 
volume of the jobs in $\opt(t',t)$ in OPT's queue.
Here, $\opt(t',t)$ denotes the set of jobs that are active in the optimal schedule
at time $t'$ and time $t$.
This significantly simplifies the proof, as we do not need to maintain a full fractional assignment of bounded expansion 
over time, but only need to argue about the volume of a \emph{single set of jobs}.
This technique has previously been used in the context of 
online flow time scheduling \cite{BansalD07,AzarT18,AzarLT21}.
A major challenge is that our SLF algorithm may process multiple 
jobs at the same time (using Round-Robin on a certain subset of jobs) and hence reduce the volume of multiple jobs at the same time.
This is a crucial difference compared to those earlier works, in which an algorithm works at any time on at most one job, and thus, makes inductive arguments about volume changes much 
easier to handle.

\section{The Shortest-Lower-Bound-First Algorithm}

We first restate the Shortest Lower-Bound First (\ALG) algorithm, which was introduced in \cite{GuptaKLSY25}.
For any time $t$, let $\ALG(t)$ denote the set of active (released and unfinished) jobs in the schedule of $\ALG$.
We further distinguish between known and unknown jobs.

\begin{defn} [Known/Unknown Jobs] 
Let $j$ be an active job in $\ALG$ at time $t$. We say that a job $j$ is \textit{known} if $r_{j}(t) \leq \oneminusA \cdot p_j$. Otherwise, $j$ is \textit{unknown}. 
\end{defn}

Let $U(t)$ and $K(t)$ denote the unknown and known jobs in the algorithm's schedule at time $t$.  
At any time $t$ and for each job $j \in \ALG(t)$, we define an \emph{estimate} $\eta_{j}(t)$
as
\[
\eta_j(t) =
\begin{cases} 
\threshold \cdot e_{j}(t), & \text{ if } j \in U(t), \\ 
r_{j}(t), & \text{ else, }
\end{cases}
\]
and \ALG works on a job $j$ with smallest estimate $\eta_j(t)$.

More specifically, \SLF behaves as follows:
\begin{enumerate}[nosep]
    \item If $\min_{j \in K(t)} \eta_j(t) \le \min_{j \in U(t)} \eta_j(t)$, \ALG processes an arbitrary known job $j$ with minimum $\eta_j(t)$.
    \item Otherwise, \ALG processes \emph{all} unknown jobs $j$ with minimum $\eta_j(t)$ in parallel at an equal pace. We can achieve this by scheduling them in Round-Robin with infinitesimally small steps.
\end{enumerate}

We assume without loss of generality that the value of $\eps$ is known to the algorithm from the start of the instance. Otherwise, the algorithm can compute it from the elapsed times once it is notified of the remaining processing time for the first job that becomes known. Before the first job becomes known, knowing the value of $\eps$ is not necessary.

\section{Analysis}

In the following, we assume that $0 < \eps < 1$. Note that if $\eps=1$, SLF is equivalent to Shortest Remaining Processing Time (SRPT), which is known to be $1$-competitive for total flow time~\cite{Schrage68}.

Our goal is to show that \ALG is locally competitive. Let $\opt(t)$ denote the set of active jobs in an optimal solution at time $t$. 

\begin{lemma}\label{lem:local-competitiveness}
	At any time $t$, it holds that $|\alg(t)| \leq \factor \cdot |\opt(t)|$.
\end{lemma}
Integrating this inequality over time yields \Cref{thm:main}.

Fix a point in time $t$ for which we want to prove \Cref{lem:local-competitiveness}. We can assume without loss of generality that until time $t$ the algorithm does not idle. Otherwise, since both the algorithm and the optimal solution do not idle unnecessarily, we can analyze every maximal time interval where the algorithm does not idle independently.

\paragraph{Terminology.}
For all times $t' \leq t$, let $\opt(t,t') := \opt(t) \cap \opt(t')$ refer to the set of jobs that are alive at time $t$ \emph{and} time $t'$ in the optimal solution. Let $\delta_{t'} := \delta(t,t') := |\opt(t,t')|$ 
denote the number of jobs that are alive at both time~$t'$ and time~$t$, in $\opt$.
Finally, we use $B(t',x)$ to refer to the $x$ alive jobs with the largest remaining volume in $\alg$ at time $t'$. 
Usually, we care only about the set $B(t',x)$ for $x = \factor \delta_{t'}$. 
Hence, we introduce the shorthand notation $B(t') := B(t',\factor \delta_{t'})$. 
For a set of jobs $Q$ and a time $t'$, let $\vol_Q(t') = \sum_{j \in Q} r_j(t')$
and $\vol^*_Q(t') = \sum_{j \in Q} r^*_j(t')$
denote the remaining volumes of jobs in $Q$ at time $t'$ in $\alg$ and $\opt$, respectively.

The algorithm's state can change instantaneously due to job arrivals and status changes (e.g., from unknown to known). We use $A(t^-)$ to denote the state of $A(t)$ just before the changes are applied at time $t$.  Intuitively, one can think of $t^{-}$ as the time instant just before $t$. The notation $t^-$ is also defined the same way for all other functions $f(t)$ and $f(t^-)$.

Our main goal is to show the following lemma. 

\begin{lemma}
    \label{lem:volume-bound}
    For all times $t' \le t$, we have
    $\vol_{B(t')}(t') \ge \vol^*_{\opt(t,t')}(t')$.
\end{lemma}

Let $\vol(t) := \vol_{\ALG(t)}(t)$ be the total remaining volume in $\alg$ at time $t$. Note that $\vol(t)$ is also equal to the total remaining volume in $\opt$ at time $t$, as neither strategy idles unnecessarily.
We can now prove~\Cref{lem:local-competitiveness}.

\begin{proof}[Proof of \Cref{lem:local-competitiveness}]
Using \Cref{lem:volume-bound} for $t' = t$, we obtain
\[
    \vol(t) \geq \vol_{B(t)}(t) \ge \vol^*_{\opt(t)}(t) = \vol(t) \ .
\]
This implies $\vol(t) = \vol_{B(t)}(t)$, meaning that the $\factor \delta_t$ jobs with the largest remaining volume in $\alg$ contain all remaining volume in $\alg$ at time $t$. Thus, there are at most $\factor \delta_t = \factor \cdot |\opt(t)|$ active jobs in $\alg$ at time $t$, which concludes the proof.
\end{proof}

It remains to prove \Cref{lem:volume-bound}, which we  do by induction over time. More specifically, we start at time $0$ and work
our way forward to time $t$. We have two different types of induction
steps to go from time $t_0$ to time $t_1 \leq t$: the Fast Forward Lemma (\Cref{lem:fastforward}) and the Suffix Carving Lemma (\Cref{lem:suffix:carving}).
We distinguish between these two cases based on the following criteria:
\begin{enumerate}
	\item If the algorithm works on unknown jobs 
	or on known jobs while there are active  unknown jobs that will
	be processed again before time $t_1$, then we use the Fast Forward Lemma (\Cref{lem:fastforward}) to go from $t_0$ to $t_1$.
	\item Otherwise, the algorithm only works on known jobs during $[t_0,t_1]$. We use the Suffix Carving Lemma (\Cref{lem:suffix:carving}) to go from $t_0$ to $t_1$. 
\end{enumerate}
Importantly, we choose \emph{maximal} intervals $[t_0,t_1]$, that is, we extend $t_1$ as far as possible while satisfying one of the two criteria above. 
To properly define those intervals, we introduce the following definitions.

\begin{defn}[Frozen jobs]
    We say a job $j$ is \emph{frozen} at time $t'$ if $j$ is unknown at $t'$ and not touched by the algorithm during $[t',t]$. We use $F(t')$ to refer to the set of frozen jobs at time $t'$. 
\end{defn}

By definition, we have $F(t_0) \subseteq F(t_1)$ for all $0 \leq t_0 \leq t_1 \leq t$.

\begin{defn}[Leader]
    Let $s \le t$ be a point in time with $|U(s)\setminus F(s)| \ge 1$. Then, define the \emph{leader} $L(s)$ at time $s$ as some job $L(s) \in U(s) \setminus F(s)$ with $e_{L(s)}(s) = \max_{j \in U(s)\setminus F(s)} e_j(s)$.
\end{defn}

The following lemma proves a useful property of frozen jobs: If some job becomes frozen at time $t'$, then all unknown jobs $j$ that were released before time $t'$ (and still active at time $t'$) become frozen at time $t'$.

\begin{lemma}\label{lem:freeze-leaders}
	If at time $t'$ a set of jobs $J_f \subseteq U(t'^-) \setminus F(t'^-)$ becomes frozen, that is, $F(t'^-) \cup J_f = F(t')$, then
    $J_f = U(t'^-) \setminus F(t'^-)$.
    In particular, the leader is touched at time $t'^-$.
\end{lemma}

\begin{proof}
Since $J_f \cap F(t'^-) = \emptyset$ and by the definition of frozen jobs,
it must be that every job in $J_f$ is touched at time $t'^-$.
Let $j_f \in J_f$.

Assume that there exists a job $j \in (U(t'^-) \setminus F(t'^-)) \setminus J_f$. This means that there exists a time $t' \leq t'' \leq t$ at which $j$ is touched. Let $t''$ be the first such time. Hence $e_j(t'^-) = e_j(t'')$. By the definition of SLF and using the fact that $j_f \in F(t'')$, it must hold that
$e_j(t'') < e_{j_f}(t'')$. 
Since $j_f$ is frozen from time $t'$ on, we have $e_{j_f}(t'^-) = e_{j_f}(t'')$.
This implies $e_{j}(t'^-) < e_{j_f}(t'^-)$. But this contradicts that $j_f$ was touched at time $t'^-$.
Hence, $J_f = U(t'^-) \setminus F(t'^-)$.
Since every job in $J_f = U(t'^-) \setminus F(t'^-)$ is touched at time $t'^-$, the leader is touched at time~$t'^-$.
\end{proof}

We use this lemma in \Cref{sec:proof-volume-bound} to prove \Cref{lem:volume-bound} at the end of this section. Before that, in the following two subsections, we state and prove the Fast Forward Lemma and the Suffix Carving Lemma.

\subsection{Fast Forward Lemma}
\label{sec:ff}

We first state the Fast Forward Lemma. 

\begin{lemma}[Fast Forward Lemma]
\label{lem:fastforward}
Let $t_0 \le t$ be such that $\vol_{B(t_0)}(t_0) \ge \vol^*_{\opt(t,t_0)}(t_0)$ and $U(t_0) = F(t_0)$, i.e., all unknown jobs at $t_0$ are frozen. Let $t_1 \ge t_0$ be such that
\begin{enumerate}[noitemsep,nolistsep]
    \item at time $t_1$, the leader 
    $L(t_1) \in U(t_1)\setminus F(t_1)$ with 
    $e_{L(t_1)}(t_1) = \max_{j \in U(t_1) \setminus F(t_1)} e_j(t_1)$ 
    is touched,
    \item the set of frozen jobs does not change during $[t_0,t_1]$,
    that is, $F(t_0) = F(t')$ for all $t' \in [t_0,t_1]$, and
    \item during $(t_0,t_1]$, there is always at least one new unknown job, that is, $|U(s)\setminus U(t_0)| \ge 1$ for all $s \in (t_0,t_1]$. This implies that there is a new job arrival during $(t_0,t_1]$.
\end{enumerate}
Then, we have that  $\vol_{B(t_1)}(t_1) \ge \vol^*_{\opt(t,t_1)}(t_1)$.
\end{lemma}

For the rest of this subsection, we fix $t_0$ and $t_1$ as described in~\Cref{lem:fastforward}. Let $J_{new}$ denote the set of jobs that are released during $(t_0,t_1)$ and let $\gamma := e_{L(t_1)}(t_1)$ denote the elapsed time of the leader $L(t_1)$ at time~$t_1$. 
We start with the following immediate corollary of Conditions~2 and~3 of \Cref{lem:fastforward}.

\begin{observation}
	\label{obs:unknown-nonfrozen}
	At any time $s \in (t_0,t_1]$, there exists an unknown non-frozen job, that is, $|U(s) \setminus F(s)| \geq 1$.
\end{observation}
\begin{proof}
	This simply follows because Condition~2 and the assumption 
	yield $F(s) = F(t_0) = U(t_0)$ for all $s \in [t_0,t_1]$.
	Then, Condition 3 yields $|U(s) \setminus F(s)| = |U(s) \setminus U(t_0)| \geq 1$ for all $s \in (t_0,t_1]$.
\end{proof}

Observe that Condition~2 of~\Cref{lem:fastforward} implies the following useful property, which states that the elapsed time of the leaders never decreases during the interval $[t_0,t_1]$. Note that the identity of the leader can change during $[t_0,t_1]$, but nevertheless, the elapsed time of the leader does not decrease.

\begin{observation}
    \label{obs:increasing:elapsed}
    The maximum elapsed time of unknown non-frozen jobs never decreases during $[t_0,t_1]$, i.e., for all $s,s' \in [t_0,t_1]$ with $s \le s'$ we have $\max_{j \in U(s') \setminus F(s')} e_j(s') \ge \max_{j \in U(s) \setminus F(s)} e_j(s)$. 
\end{observation}

\begin{proof}
    For the sake of contradiction, assume that there are points in time $s,s' \in [t_0,t_1]$ with $s \le s'$ but $\max_{j \in U(s') \setminus F(s')} e_j(s') < \max_{j \in U(s) \setminus F(s)} e_j(s)$. Then, there must be some time $s \le \bar{s} \le s'$ with  $\max_{j \in U(\bar{s}^+) \setminus F(\bar{s}^+)} e_j(\bar{s}^+) < \max_{j \in U(\bar{s}) \setminus F(\bar{s})} e_j(\bar{s})$. This can only happen if the leader $L(\bar{s})$ and all jobs in $U(\bar{s})\setminus F(\bar{s})$ with the same elapsed time as the leader become known or frozen at $\bar{s}$. In either case, the leader is touched at $\bar{s}$, which implies that all jobs in $U(\bar{s})\setminus F(\bar{s})$ have the same elapsed time as the leader. Hence, all jobs in $U(\bar{s})\setminus F(\bar{s})$ become known or frozen at time $\bar{s}$. This is a contradiction to \Cref{obs:unknown-nonfrozen}.
\end{proof}

To prove~\Cref{lem:fastforward}, we rely on the following observation. This corresponds to \cite[Lemma~V.13]{GuptaKLSY25}. Compared to \cite{GuptaKLSY25}, the argument slightly changes because we need to use the second assumption of~\Cref{lem:fastforward} instead of the early arriving property. 

\begin{lemma}[Switching]
    \label{obs:switching}
The following statements hold: 
\begin{enumerate}[noitemsep,nolistsep,label=(\roman*)]
        \item For every job $j \in K(t_0)$, $j \in \ALG(t_1)$ if and only if $r_{j}(t_0) >    \threshold\cdot \gamma$.
        \item If $ j\in K(t_1)$, then $r_{j}(t_0) = r_{j}(t_1).$
Thus, $K(t_1) \subseteq K(t_0)$ and $\ALG$ does not touch $j$ during $(t_0,t_1]$.
        \item For every job $j \in U(t_1) \setminus F(t_1)$, we have $e_{j}(t_1) = \gamma$. \item For every job $j \in J_{new} \setminus \ALG(t_1)$, we have $p_j = e_{j}(t_1) \leq \frac{1}{1-\varepsilon} \cdot  \gamma$.
    \end{enumerate}
\end{lemma}

\begin{proof}
    We separately prove the four statements:
\begin{enumerate}[label=(\roman*)]
    \item Fix an arbitrary job $j \in K(t_0) \cap \ALG(t_1)$. Since $j$ is alive and known at $t_1$ but the leader $L(t_1)$ is touched at $t_1$, we must have $r_j(t_0) > \threshold \cdot e_{L(t_1)}(t_1) = \threshold \gamma$ by  the definition of $\ALG$. 

    Fix an arbitrary job $j \in K(t_0) \setminus \ALG(t_1)$. Let $s$ denote the earliest point in time during $(t_0,t_1)$ at which $j$ is touched. 
    Using~\Cref{obs:increasing:elapsed}, we get
    $r_j(t_0) = r_j(s) \le \threshold e_{L(s)}(s) \le \threshold e_{L(t_1)}(t_1) = \threshold \gamma$.
    
    \item We first argue that there cannot be a job $j \in J_{new} \cap K(t_1)$. If there was such a job $j$, then there is a point in time $t_0 < s < t_1$ at which $j$ becomes known. By 
\Cref{obs:unknown-nonfrozen},
    there needs to be at least another non-frozen unknown job $j' \in U(s) \setminus F(s)$. Consider the leader $L(s)$ among those jobs ($e_{L(s)}(s) = \max_{j \in U(s) \setminus F(s)} e_j(s)$). Note that we might have $L(s) \not= L(t_1)$.
    Since $j$ becomes known at $s$, we have $r_j(s) = \threshold e_j(s) \le \threshold e_{L(s)}(s)$. Using that the remaining size of $j$ can only decrease over time and that the elapsed time of the (potentially changing) leader only increases over time by~\Cref{obs:increasing:elapsed}, we get $r_j(t_1) \le \threshold \gamma$. If $r_j(t_1) > 0$, then this is a contradiction to $L(t_1)$ being touched at $t_1$. Hence, $j \not\in \ALG(t_1)$.

    The argument for jobs in $K(t_0)$ is essentially the same. If $r_j(t_0) > r_j(t_1)$, then $j$ is touched at some point in time $t_0 < s < t_1$. Hence, $r_j(s) \le \threshold e_{L(s)}(s)$, where the leader $L(s)$ exists by the proof of \Cref{lem:volume-bound}.
Using~\Cref{obs:increasing:elapsed}, we get $r_j(t_1) \le \threshold \gamma$, a contradiction to $L(t_1)$ being touched at $t_1$.

    \item Since the leader $L(t_1)$ with $e_{L(t_1)}(t_1) = \gamma$ is touched at $t_1$, all $j \in U(t_1) \setminus F(t_1)$ need to satisfy $e_j(t_1) = \gamma$ by the definition of $\ALG$.

    \item Consider a job in $j \in J_{new} \setminus \ALG(t_1)$ and let $s \in (t_0,t_1)$ denote the time at which $j$ becomes known. Since $j$ is touched at time $s$, we must have $e_j(s) \le e_{L(s)}(s)$. Using~\Cref{obs:increasing:elapsed}, this gives $e_j(s) \le e_{L(t_1)}(t_1) = \gamma$.  Finally, we can observe that $j$ becoming known at $s$ implies $e_j(s) = (1-\varepsilon) p_j$. Hence, $e_j(t_1) = p_j = \frac{e_j(s)}{1-\varepsilon} \le \frac{1}{1-\varepsilon} \cdot \gamma$. \qedhere
\end{enumerate}
\end{proof}

Recall that $B(t_0)$ is the set of the $\factor \delta_{t_0}$ largest jobs in $\ALG(t_0)$. Let $D = B(t_0) \setminus \ALG(t_1)$ denote the (possibly empty) set of jobs in $B(t_0)$ that $\ALG$ completes by time $t_1$.  Let $O^+ = \OPT(t,t_1) \setminus \OPT(t,t_0)$ denote the new jobs that enter $\OPT(t,t_1)$.

Our strategy for proving~\Cref{lem:fastforward} is to construct a set $S \subseteq \ALG(t_1)$ with
    \begin{enumerate}[noitemsep,nolistsep]
        \item $|S| \le \factor \delta_{t_0} +  \factor \cdot |O^+| = \factor \delta_{t_1}$ and
        \item $\vol_S(t_1) \ge \vol^*_{\OPT(t,t_0)}(t_1) + \vol^*_{O^+}(t_1) = \vol^*_{\OPT(t,t_1)}(t_1)$.
    \end{enumerate}

    If we find such a set $S$, then we immediately get $\vol_{B(t_1)}(t_1) \ge \vol_S(t_1) \ge \vol^*_{\OPT(t,t_1)}(t_1)$, and we are done. We now construct such a set $S$:
    \begin{enumerate}
        \item Let $S_1 := (B(t_0) \setminus D) \cup (O^+ \cap \ALG(t_1))$.
        \item Let $S_2$ denote an arbitrary subset of $\ALG(t_1) \setminus S_1$ with $$|S_2| = \ceil{\frac{1}{\varepsilon}} \cdot  |O^+ \setminus \ALG(t_1)| + |D| + \left(\ceil{\frac{1}{\eps}} -1\right) \cdot |O^+ \cap \ALG(t_1)|.$$ 
        If no such subset exists, then let $S_2 = \ALG(t_1) \setminus S_1$. 
        \item Define $S := S_1 \cup S_2$.
    \end{enumerate}
    It remains to show that $S$ satisfies the two requirements above. We first show that $S$ is sufficiently small.
    \begin{lemma}
    \label{obs:ff:cardinality}
        It holds that
        $|S| \le \factor \delta_{t_0} +  \factor |O^+| = \factor \delta_{t_1}$.
    \end{lemma}

    \begin{proof}
    Using the definition of $S$, we can conclude
        \begin{align*}
            |S| &\le |B(t_0)\setminus D| +|O^+ \cap \ALG(t_1)| + \ceil{\frac{1}{\varepsilon}} \cdot |O^+ \setminus \ALG(t_1)| + |D| + \left(\ceil{\frac{1}{\varepsilon}}-1\right) \cdot |O^+ \cap \SLF(t_1)|\\
            &= |B(t_0)| + \ceil{\frac{1}{\varepsilon}} \cdot  |O^+ \setminus \ALG(t_1)| + \ceil{\frac{1}{\varepsilon}} \cdot |O^+ \cap \ALG(t_1)|\\
            &= |B(t_0)| + \ceil{\frac{1}{\varepsilon}} \cdot |O^+|\\
            &\le \ceil{\frac{1}{\varepsilon}} \delta_{t_0} + \ceil{\frac{1}{\varepsilon}} \cdot |O^+| = \ceil{\frac{1}{\varepsilon}} \delta_{t_1} \ ,
        \end{align*}
        where the last inequality follows from the definition of $B(t_0)$, and the last equality follows from the definition of $\delta_{t_1}$.
        This completes the proof.
    \end{proof}

    It remains to show that the volume of $S$ is sufficiently large, i.e., $\vol_S(t_1) \ge \vol^*_{\OPT(t,t_0)}(t_1) + \vol^*_{O^+}(t_1)$. First, note that if $\SLF(t_1) \setminus S_1$ is too small in the second step of the construction of $S$, then the statement holds trivially.

    \begin{lemma}
        \label{obs:ff:fewelements}
        If $S_2 = \SLF(t_1)\setminus S_1$, then $\vol_S(t_1) \ge \vol^*_{\OPT(t,t_1)}(t_1)$.
    \end{lemma}

    \begin{proof}
        If $S_2 = \SLF(t_1)\setminus S_1$, then $S = \SLF(t_1)$ and, thus, $\vol_{S}(t_1) = \vol_{\SLF(t_1)}(t_1)$. Using the assumption that the machine never idles, this gives us
        $\vol_{S}(t_1) = \vol_{\SLF(t_1)}(t_1) = \vol^*_{\OPT(t_1)}(t_1) \ge \vol^*_{\OPT(t,t_1)}(t_1)$.
    \end{proof}    

    So we assume $|S_2| < |\SLF(t_1)\setminus S_1|$ for the rest of the proof. We continue by giving two separate lower bounds on the volumes $\vol_{S_1}(t_1)$ and $\vol_{S_2}(t_1)$ in the following \Cref{lem:vol:lb:1,lem:vol:lb:2}.

    \begin{lemma}
    \label{lem:vol:lb:1}
It holds that
    $$\vol_{S_1}(t_1) \ge \vol^*_{\OPT(t,t_0) \cup O^+}(t_1) - \gamma \cdot |O^+ \cap \SLF(t_1)| - \frac{1}{1-\varepsilon} \cdot \gamma \cdot |O^+ \setminus \SLF(t_1)| - \threshold \cdot \gamma \cdot |D| \ .$$
    \end{lemma}

    \begin{proof}
        The lemma follows from the following chain of inequalities. We argue  below that these inequalities indeed hold:
        \begin{align}
            \vol_{S_1}(t_1) &=\vol_{B(t_0)\setminus D}(t_1) + \vol_{O^+ \cap \ALG(t_1)}(t_1) \nonumber\\
            &= \vol_{B(t_0)\setminus D}(t_0) + \vol_{O^+ \cap \ALG(t_1)}(t_1)  \label{volS2:1}\\
            &\ge \vol^*_{\OPT(t,t_0)}(t_1) - \vol_D(t_0) + \vol_{O^+ \cap \ALG(t_1)}(t_1) \label{volS2:2}\\
            &\ge \vol^*_{\OPT(t,t_0)}(t_1) - \vol_D(t_0) + \vol^*_{O^+ \cap \SLF(t_1)}(t_1) -  \gamma \cdot |O^+ \cap \SLF(t_1)|  \label{volS2:3}\\
            &= \vol^*_{\OPT(t,t_0)}(t_1) + \vol^*_{O^+}(t_1) -  \vol_D(t_0) - \gamma \cdot|O^+ \cap \SLF(t_1)|  - \vol^*_{O^+ \setminus \SLF(t_1)}(t_1)\nonumber\\
            &\ge \vol^*_{\OPT(t,t_0) \cup O^+}(t_1) -  \vol_D(t_0) - \gamma\cdot |O^+ \cap \SLF(t_1)|  - \frac{1}{1-\varepsilon}  \gamma \cdot |O^+ \setminus \SLF(t_1)| \label{volS2:4}\\
            &\ge \vol^*_{\OPT(t,t_0) \cup O^+}(t_1) -  \gamma\cdot |O^+ \cap \SLF(t_1)|  - \frac{1}{1-\varepsilon} \cdot \gamma \cdot |O^+ \setminus \SLF(t_1)| - \threshold \cdot \gamma \cdot |D|.
            \label{volS2:5}
        \end{align}
    We separately argue that the equations and inequalities above indeed hold:
    \begin{itemize}
        \item Equation~\eqref{volS2:1} holds because each $j \in B(t_0) \setminus D$ is either in $K(t_0)$ or in $F(t_0)$ by assumption that $U(t_0)=F(t_0)$ of~\Cref{lem:fastforward}. In either case, we have $r_j(t_0) = r_j(t_1)$. If $j \in K(t_0)$, then $r_j(t_0) = r_j(t_1)$ holds by \Cref{obs:switching}(ii). If $j \in F(t_0)$, then  $r_j(t_0) = r_j(t_1)$ holds by the definition of frozen jobs. The fact that $r_j(t_0) = r_j(t_1)$ for all $j \in B(t_0) \setminus D$ implies $\vol_{B(t_0)\setminus D}(t_1) = \vol_{B(t_0)\setminus D}(t_0)$.
        \item Inequality~\eqref{volS2:2} uses that $\vol_{B(t_0)}(t_0) \ge \vol^*_{\OPT(t,t_0)}(t_0)$ holds by the assumption of~\Cref{lem:fastforward}, and since $\vol^*_{\OPT(t,t_0)}(t_0) \ge \vol^*_{\OPT(t,t_0)}(t_1)$.
        \item Inequality~\eqref{volS2:3} exploits that $e_j(t_1) = \gamma$ for each $j \in J_{new} \cap \SLF(t_1)$.
        To see this, first observe that 
        $K(t_1) \subseteq K(t_0)$ by \Cref{obs:switching}(ii).
        Since $K(t_0) \cap J_{new} = \emptyset$, $F(t_0) \cap J_{new} = \emptyset$, and by Condition 2 of \Cref{lem:fastforward}, it must be that $J_{new} \cap \SLF(t_1) \subseteq U(t_1) \setminus F(t_1)$. 
        Thus, \Cref{obs:switching}(iii) gives that $e_j(t_1) = \gamma$ for each $j \in J_{new} \cap \SLF(t_1)$.
        Hence,
        \begin{align*}
           \vol^*_{O^+ \cap \SLF(t_1)}(t_1) &\le \sum_{j \in O^+ \cap \SLF(t_1)} p_j \\
           &= \sum_{j \in O^+ \cap \SLF(t_1)} r_j(t_1) + e_j(t_1) \\
           &= \vol_{O^+ \cap \SLF(t_1)}(t_1) + |O^+ \cap \SLF(t_1)| \cdot \gamma.
        \end{align*}
        Therefore,
        $\vol_{O^+ \cap \SLF(t_1)}(t_1) \ge \vol^*_{O^+ \cap \SLF(t_1)}(t_1) - |O^+ \cap \SLF(t_1)| \cdot \gamma$.
        \item Inequality~\eqref{volS2:4} follows since each $j \in O^+ \setminus \ALG(t_1)$ satisfies $j \in J_{new} \setminus \ALG(t_1)$ by definition of~$O^+$. Then, \Cref{obs:switching}(iv) implies $p_j \le\frac{1}{1-\varepsilon} \gamma$ for all such $j \in J_{new} \setminus \ALG(t_1)$. Thus,
        $$
        \vol^*_{O^+\setminus \ALG(t_1)}(t_1) \le \sum_{j \in O^+\setminus \ALG(t_1)} p_j \le |O^+\setminus \ALG(t_1)| \cdot \frac{1}{1-\varepsilon} \cdot \gamma.
        $$
        \item Finally, Inequality~\eqref{volS2:5} follows since each $j \in D$ is in $K(t_0) \setminus \SLF(t_1)$. Hence, \Cref{obs:switching}(i) implies $r_j(t_0) \le \threshold \cdot \gamma$, and we can conclude with
        $$
        \vol_D(t_0) = \sum_{j \in D} r_j(t_0) \le \threshold \cdot \gamma \cdot |D|.
        $$
    \end{itemize}
    This completes the proof of the lemma.
    \end{proof}

    Having~\Cref{lem:vol:lb:1}, we only need to show that the volume of $S_2$ is large enough to cover the remaining volume. This is formalized in the following lemma.

    \begin{lemma}
        \label{lem:vol:lb:2}
        $\vol_{S_2}(t_1) \ge \gamma \cdot |O^+ \cap \SLF(t_1)| +\frac{1}{1-\varepsilon} \cdot \gamma \cdot |O^+ \setminus \SLF(t_1)| + \threshold \cdot \gamma \cdot |D|$
    \end{lemma}

    \begin{proof}
        Recall that $S_2$ is a subset of $\SLF(t_1)\setminus S_1$ with $|S_2| = \factor \cdot  |O^+ \setminus \ALG(t_1)| + |D| + \left(\factor -1\right) \cdot |O^+ \cap \ALG(t_1)|$. Next, observe that each $j \in \SLF(t_1)$ satisfies  
        $$
            r_j(t_1) \ge \threshold \cdot  \gamma.
        $$
        For $j \in U(t_1)$, the fact that $j$ is unknown at $t_1$ and~\Cref{obs:switching}(iii) imply $r_j(t_1) \ge  \threshold e_j(t_1) = \threshold \gamma$.
        For $j \in K(t_1)$,~\Cref{obs:switching}(ii) implies $j \in K(t_0)$  and, therefore, $r_j(t_1) = r_j(t_0) > \threshold \gamma$ (\Cref{obs:switching}(i)).

        Partition $S_2$ into three sets $A_1$, $A_2$ and $A_3$ with 
        $|A_1| = \factor |O^+ \setminus \SLF(t_1)|$ , $|A_2| = |D|$, and $|A_3| = \left(\factor -1\right) \cdot |O^+ \cap \ALG(t_1)|$. Then,
         \begin{align*}
            \vol_{S_2}(t_1) &= \vol_{A_1}(t_1) + \vol_{A_2}(t_1) + \vol_{A_3}(t_1)\\
            &\ge \gamma \cdot \threshold \cdot \ceil{\frac{1}{\varepsilon}} \cdot  |O^+ \setminus \ALG(t_1)| + \gamma \cdot \threshold \cdot  |D| + \gamma \cdot \threshold \cdot  \left(\ceil{\frac{1}{\eps}} -1\right) \cdot |O^+ \cap \ALG(t_1)|\\
            &\ge \gamma \cdot \frac{1}{1-\varepsilon} \cdot |O^+ \setminus \ALG(t_1)|  + \gamma \cdot \threshold \cdot |D| + \gamma  \cdot |O^+ \cap \ALG(t_1)| \ ,
        \end{align*}
        since $\frac{\eps}{1-\eps} \cdot  \left(\factor -1\right) \ge 1$ for every $0< \varepsilon < 1$. 
    \end{proof}

    Putting everything together, we can finish the proof of the Fast Forward \Cref{lem:fastforward}.

    \begin{proof}[Proof of~\Cref{lem:fastforward}]
        By~\Cref{obs:ff:cardinality}, the set $S \subseteq \SLF(t_1)$ as constructed above satisfies $|S| \le \factor  \delta_{t_1}$. The~\Cref{obs:ff:fewelements} and~\Cref{lem:vol:lb:1,lem:vol:lb:2} imply $\vol_{S}(t_1) \ge \vol^*_{\OPT(t,t_1)}(t_1)$.
        Since $B(t_1)$ is the set of the $\factor  \delta_{t_1}$ largest elements in $\SLF(t_1)$, we conclude $\vol_{B(t_1)}(t_1) \ge \vol_{S}(t_1) \ge \vol^*_{\OPT(t,t_1)}(t_1)$.
    \end{proof}

\subsection{Suffix Carving Lemma}
\label{sec:suffix-carving}

We next state and prove the Suffix Carving Lemma, which we can use for time intervals when SLF only works on known jobs.

    \begin{lemma}[Suffix Carving]
    \label{lem:suffix:carving}
    Let $t_0 < t_1 \le t$ be two points in time such that
    \begin{enumerate}[noitemsep,nolistsep]
        \item $\vol_{B(t_0)}(t_0) \ge \vol^*_{\OPT(t,t_0)}(t_0)$, and
        \item $\SLF$ only works on known jobs during $[t_0,t_1]$.
    \end{enumerate}
    Then, $\vol_{B(t_1)}(t_1) \ge \vol^*_{\OPT(t,t_1)}(t_1)$.
    \end{lemma}

    \begin{proof}
    The assumption that $\SLF$ works only on known jobs during $[t_0,t_1]$ implies that no new jobs are released during that interval. Then, $\OPT(t,t_0) = \OPT(t,t_1)$, and thus, $\vol^*_{\OPT(t,t_0)}(t_0) \ge \vol^*_{\OPT(t,t_0)}(t_1) \ge \vol^*_{\OPT(t,t_1)}(t_1)$. 

    If $\SLF$ does not touch jobs in $B(t_0)$ during $[t_0,t_1]$, then
    $$
    \vol_{B(t_1)}(t_1)  = \vol_{B(t_0)}(t_0) \ge \vol^*_{\OPT(t,t_0)}(t_0) \ge  \vol^*_{\OPT(t,t_1)}(t_1),
    $$
    by the first assumption of the lemma, and we are done.

    If $\SLF$ at some point touches a job  $j\in B(t_0)$, then, using the fact  that when $\SLF$ works on a known job 
    then it is the smallest in its queue, we conclude that $j$ is the smallest in SLF's queue. This implies that the total number of jobs in $\SLF(t_1)$ is at most $|B(t_0)| \le \factor \delta_{t_0} = \factor \delta_{t_1}$. Hence, $B(t_1) = \SLF(t_1)$. Using the assumption that the machine never idles, we have
    $$
        \vol_{B(t_1)}(t_1) = \vol_{\SLF(t_1)}(t_1) = \vol^*_{\OPT(t_1)}(t_1) \ge \vol^*_{\OPT(t,t_1)}(t_1) \ ,
    $$
    which concludes the proof of the lemma.
    \end{proof}

\subsection{Proof of~\Cref{lem:volume-bound}}
\label{sec:proof-volume-bound}

We prove the lemma by induction over time $t'$ and 
    maintain the invariant that 
    \[
    \vol_{B(t')}(t') \ge \vol^*_{\opt(t,t')}(t') \ .
   \] 
    Since we assume that SLF does not idle until time $t$, the invariant holds trivially for time $0$.
    We now iteratively verify the invariant for times $t_1$ with 
    $t_0 < t_1 \leq t$ until we reach time $t$; 
    see \Cref{fig:induction-proof} for an example.
 Initially $t_0 = 0$. Then we define $t_1$ according to one of the following two cases. For the next iteration, we set $t_0 := t_1$ and repeat until $t_1=t$. It will be clear that $t_0$ always satisfies one of the following conditions. 

\begin{figure}[tb]
    \centering
	\begin{tikzpicture}[xscale=0.65]
	\rectjobT[job1](0,0)(3,0.5)($1$);
	\rectjobT[job2](0,0.5)(3,0.5)($2$);
	\rectjobT[job1](3,0)(2,1)($1$);
	\rectjobT[job2](5,0)(2,1)($2$);
	\rectjobT[job2](7,0)(2,1)($2$);
	\rectjobT[job3](9,0)(2,0.5)($3$);
	\rectjobT[job4](9,0.5)(2,0.5)($4$);
	\rectjobT[job2](11,0)(1,1)($2$);
	\rectjobT[job3](12,0)(2,0.5)($3$);
	\rectjobT[job4](12,0.5)(2,0.5)($4$);
	\rectjobT[job3](14,0)(2,1)($3$);
	\rectjobT[job4](16,0)(2,1)($4$);
	\rectjobT[job5](18,0)(3,1)($5$);
	\rectjobT[job6](21,0)(2,0.5)($6$);
	\rectjobT[job7](21,0.5)(2,0.5)($7$);
	\rectjobT[job6](23,0)(1,1)($6$);
	\rectjobT[job7](24,0)(1,1)($7$);
	
	\begin{scope}[shift={(0,-0.5)}]	
	\draw[thick,->] (0,0) -- (26,0);
	\draw[thick] (0,-0.15) --node[anchor=south,below=4pt] (t0) {$\strut 0$}  (0,0.15);
	\draw[thick] (3,-0.15) --node[anchor=south,below=4pt] (t1) {$\strut s_1$}  (3,0.15);
	\draw[thick] (5,-0.15) --node[anchor=south,below=4pt] (t2) {$\strut C_1$}  (5,0.15);
	\draw[thick] (7,-0.15) --node[anchor=south,below=4pt] (t3) {$\strut s_2$}  (7,0.15);
	\draw[thick] (9,-0.15) --node[anchor=south,below=4pt] (t4) {$\strut r_{3},r_4$}  (9,0.15);
	\draw[thick] (12,-0.15) --node[anchor=south,below=4pt] (t5) {$\strut C_2$}  (12,0.15);
	\draw[thick] (14,-0.15) --node[anchor=south,below=4pt] (t6) {$\strut s_3$}  (14,0.15);
	\draw[thick] (16,-0.15) --node[anchor=south,below=4pt] (t7) {$\strut C_3$}  (16,0.15);
	\draw[thick] (18,-0.15) --node[anchor=south,below=4pt] (t8) {$\strut r_5$}  (18,0.15);
	\draw[thick] (21,-0.15) --node[anchor=south,below=4pt] (t9) {$\strut r_{6},r_7$}  (21,0.15);
	\draw[thick] (23,-0.15) --node[anchor=south,below=4pt] (t10) {$\strut s_6$}  (23,0.15);
	\draw[thick] (24,-0.15) --node[anchor=south,below=4pt] (t11) {$\strut C_6$}  (24,0.15);
	\draw[thick] (25,-0.15) --node[anchor=south,below=4pt] (t12) {$\strut t$}  (25,0.15);
	\end{scope}
	
	\draw[thick,->] (t0) to [bend right=20] node[below] {FF, Case 2(b)} (t3);
	\draw[thick,->] (t3) to [bend right=30] node[below] {SC, Case 1}  (t4);
	\draw[thick,->] (t4) to [bend right=20] node[below] {FF, Case 2(a)} (t8);
	\draw[thick,->] (t8) to [bend right=30] node[below] {FF, Case 2(a)}  (t9);
	\draw[thick,->] (t9) to [bend right=30]node[below] {FF, Case 2}   (t12);
	
	\machine[](0,0)(26,1);
	\end{tikzpicture}
	\caption{Schedule of SLF and corresponding inductive proof of \Cref{lem:volume-bound}. The arrows indicate the application of either the Fast Forward Lemma (FF) or the Suffix Carving Lemma (SC) at the respective time intervals, and the corresponding case in the proof. Here, $s_j$ denotes the point in time when job $j$ becomes known, $r_j$ denotes the arrival time of job $j$, and $C_j$ denotes the completion time of job $j$.}
	\label{fig:induction-proof}
\end{figure}
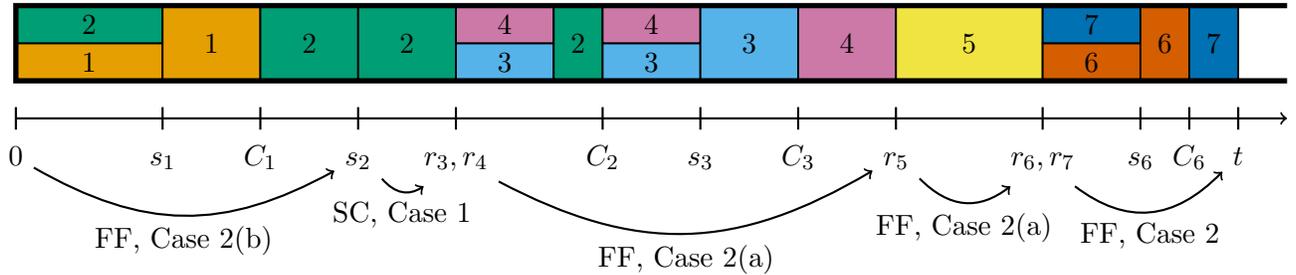

    \begin{enumerate}
        \item If at time $t_0$ all unknown jobs are frozen, that is, $|U(t_0) \setminus F(t_0)| = 0$, then let $t_1$ be the earliest point in time with $t_0 < t_1 \le t$ where there exists at least one unknown non-frozen job, that is, $|U(t_1) \setminus F(t_1)| \geq 1$. 
        If no such point in time exists, then let $t_1 = t$. Since there are no non-frozen unknown jobs during $[t_0,t_1)$ and $\SLF$ does not idle, we know that $\SLF$ only works on known jobs during this interval. Thus, we can apply the Suffix Carving \Cref{lem:suffix:carving} to conclude that the invariant is satisfied at $t_1$. Moreover, note that $t_1$ satisfies the condition of Case 2. So in the next iteration when we set $t_0 = t_1$ we can apply Case 2.
        \item If at $t_0$ an unknown non-frozen job arrives, that is, $|U(t_0) \setminus F(t_0)| \ge 1$ and $|U(t_0^-) \setminus F(t_0^-)| = 0$, 
        then let $t_1$ be the earliest point in time with $t_0 < t_1 \le t$ 
        at which either
        a job becomes frozen, i.e., $F(t_1^-) \neq F(t_1)$
        or
        all unknown jobs are frozen, i.e.,
        $U(t_1) = F(t_1)$ (because jobs became known).
        If such a time $t_1$ exists, it must be $F(t_0) = F(s)$ and $U(s) \setminus U(t_0) \neq \emptyset$ for all $s \in (t_0,t_1]$.
        We now distinguish two cases: 
        \begin{enumerate}
            \item A job becomes frozen at time $t_1$, i.e., $F(t_1^-) \neq F(t_1)$. 
            This happens when a new job arrives at time~$t_1$.
            Let $J_f = F(t_1) \setminus F(t_1^-)$ be the set of jobs that become frozen at time~$t_1$. 
            By \Cref{lem:freeze-leaders}, 
            we have $J_f = U(t_1^-) \setminus F(t_1^-)$ and 
            the leader was touched at time~$t_1^-$.
            Thus, we can apply~\Cref{lem:fastforward} at time~$t_0^-$ to conclude that the invariant is satisfied at $t_1^-$.  Moreover, time~$t_1$ satisfies the condition of Case~2.
            \item The set of frozen jobs remains, but all unknown non-frozen jobs become known at time~$t_1$, that is, $U(t_1) = F(t_1)$ and $U(t_1^-) \setminus F(t_1^-) \neq \emptyset$.
            By the definition of frozen jobs, this means that every job in $U(t_1^-) \setminus F(t_1^-)$ is touched at time~$t_1^-$. In particular, $L(t_1^-)$ is touched at time~$t_1^-$.
            Thus, we can apply~\Cref{lem:fastforward} at time~$t_0^-$ to conclude that the invariant is satisfied at $t_1^-$. Moreover, time~$t_1$ satisfies the condition of Case~1.
        \end{enumerate}      
        If no such point $t_1$ in time exists, then we have $|U(t) \setminus F(t)| \geq 1$. By the definition of frozen jobs, 
        every job that is touched at time $t$ is non-frozen. Since we are not in case 1, the known jobs are not processed at time $t$. Thus, the leader is touched at time $t$.
        This allows us to apply~\Cref{lem:fastforward} at time $t_0^-$ to conclude that the invariant is satisfied at time $t$.
        \end{enumerate}
        This shows that the invariant is satisfied at time $t$, which concludes the proof of the lemma.
        \qed

\printbibliography

\end{document}